\documentclass[10pt, twocolumn]{IEEEtran}

\usepackage{xspace,amsmath,amssymb,amsfonts,epsfig,syntonly}
\usepackage{cite,bm,color,url,textcomp}
\usepackage{algorithmicx,algorithm,algpseudocode}
\usepackage{epstopdf}
\usepackage{empheq}
\usepackage{flushend}
\usepackage{subfigure}

\newtheorem{lemma}{\hspace{-11pt}\bf Lemma}

\newtheorem{proposition}{\hspace{-11pt}\bf Proposition}
\newtheorem{theorem}{\hspace{-11pt}\bf Theorem}

\newtheorem{remark}{\hspace{-11pt}\bf Remark}

\newtheorem{property}{\hspace{-11pt}\bf Property}

\long\def\symbolfootnote[#1]#2{\begingroup
\def\thefootnote{\fnsymbol{footnote}}
\footnote[#1]{#2}\endgroup}

\psfull

\allowdisplaybreaks[3]

\begin{document}
\title{Stochastic Online Control for Energy-Harvesting Wireless Networks with Battery Imperfections}
\author{Xin Wang,~\IEEEmembership{Senior Member, IEEE}, Tianhui Ma,
Rongsheng Zhang, and Xiaolin Zhou,~\IEEEmembership{Member, IEEE}
\thanks {Work in this paper was supported by the National Natural Science Foundation
of China under Grant No. 61571135, the China Recruitment Program of Global Young Experts, the Program for New Century Excellent Talents in University, the Innovation Program of Shanghai Municipal Education Commission.}
\thanks{X. Wang, T. Ma, R. Zhang, and X. Zhou are with the Key Laboratory for Information Science of Electromagnetic Waves (MoE), Department of Communication Science and Engineering, Fudan University, 220 Han Dan Road, Shanghai, China. Email: \{xwang11, 14210720145, 15210720130, zhouxiaolin\}@fudan.edu.cn.}}
%
\markboth{}{}
\maketitle
\begin{abstract}
In energy harvesting (EH) network, the energy storage devices (i.e., batteries) are usually not perfect. In this paper, we consider a practical battery model with finite battery capacity, energy (dis-)charging loss, and energy dissipation.
Taking into account such battery imperfections, we rely on the Lyapunov optimization technique to develop a stochastic online control scheme that aims to maximize the utility of data rates for EH multi-hop  wireless networks. It is established that the proposed algorithm can provide a feasible and efficient data admission, power allocation,  routing and scheduling solution, without requiring any statistical knowledge of the stochastic channel, data-traffic, and EH processes.
Numerical results demonstrate the merit of the proposed scheme.
\end{abstract}


\section{Introduction}

From being a scientific curiosity only a few years ago, energy harvesting (EH) is well on its way to becoming a game-changing technology in the field of self-sustainable, autonomous wireless networked systems. A major factor that has contributed to the growth of the EH market is the evolution of ultra-low power electronics, which can run on the minuscule amounts of energy supplied by typical solar, vibration or thermal energy harvesters \cite{Sud11}. A number of companies are already offering system solutions consisting exclusively of EH sensor nodes \cite{Lord, Tec11}. Stimulated by these advances, EH-powered wireless communications have attracted growing interest in recent years \cite{Lei09, Sha10}.

Different from traditional communication systems, EH from environmental sources shifts the paradigm on resource allocation from reducing energy consumption to the most efficient utilization of opportunistic energy. Energy availability constraints are imposed such that the energy accumulatively consumed up to any time cannot exceed what has been accumulatively harvested so far. Considering these new type of constraints, optimal transmission policies were characterized for point-to-point channels in \cite{Yan12, Ho12, Tut12, Oze11, Lei09, Sha10}, for broadcast channels in \cite{Ant11, Yang12, Oze12}, for multi-access channels in \cite{YANG12}, for interference channels in \cite{TUT12}, for two-hop relay channels in \cite{Hua13, YLuo13}, for transmitters with non-ideal circuit-power consumption \cite{Xu13, Orh12, Oze13, Wan15}, and for systems with battery imperfections in \cite{Dev12, Luo13}.

Existing works \cite{Yan12, Ho12, Tut12, Oze11, Lei09, Sha10, Ant11, Yang12, Oze12, YANG12, TUT12, Hua13, YLuo13, Xu13, Orh12, Oze13, Wan15, Dev12, Luo13} on EH communications mostly addressed offline optimizations, where the EH profiles were assumed to be known {\em a priori}. In practical scenarios, complete predictability of EH profiles is clearly an over-simplified assumption. Relying on past realizations of EH processes and certain statistics of their future evolutions, \cite{YLuo13, Oze11, Xu13, Wan15} developed some heuristic online algorithms, which, however, lack strong analytical performance guarantees. By modeling the EH and/or data processes as Markov processes, online optimizations were cast as Markov decision problems (MDP) and numerically solved with dynamic programming tools in \cite{Bla13, Lei09, Sha10, Sri13}. However, the well-known ``curse-of-dimensionality'' with such solutions precludes their application for all but the simplest practical networks.

Leveraging stochastic optimization tools, a few low-complexity online schemes were developed in \cite{Lin07, Gat10, LHua13, Mao12, Che14}. Considering the energy-aware routing with energy replenishment at nodes, \cite{Lin07} proposed an algorithm that achieves a logarithmic competitive ratio and is asymptotically optimal as the network size grows. A simple asymptotically optimal joint energy allocation and routing scheme was developed for rechargeable sensor networks with static and non-interfering (i.e., orthogonal) links in \cite{Che14}. Relying on the Lyapunov optimization techniques, \cite{Gat10, LHua13} developed and analyzed the utility-optimal resource scheduling schemes for general EH-powered multi-hop wireless networks, while \cite{Mao12} addressed the effect of finite energy and data storage capacities on such resource allocation tasks.

The schemes in \cite{Lin07, Gat10, LHua13, Mao12, Che14} assumed ideal energy storage devices (i.e., batteries) in use. Under this assumption, the energy-queue sizes at the batteries can play the role of ``stochastic'' Lagrange multipliers to develop a dual-subgradient based solver to the intended problems. However, the imperfections with practical batteries could disable this approach. In this paper, we consider a practical battery model accounting for finite battery capacity, energy (dis-)charging loss, and energy dissipation over time. By integrating and generalizing  the Lyapunov optimization techniques in \cite{LHua13, Qin15}, we re-establish a systematic framework to develop and analyze the stochastic online control schemes for EH wireless networks with such imperfect batteries. Specifically, we propose a data-backpressure based scheduling and degenerated energy-queue based power allocation scheme that can maximize the utility of data rates for EH multi-hop wireless networks, without requiring any statistical knowledge of the stochastic channel, data-traffic, and EH processes. Different from \cite{LHua13} where an EH admission mechanism is performed to ensure finite energy queues, we apply the sample path analysis in \cite{Urg11, WangJSAC15} to derive the conditions that the proposed scheme is feasible for any given finite battery capacities without EH admission, which can help fully exploit the available harvested energy. In addition, we rigorously establish the performance guarantees of the proposed scheme in form of sub-optimality bounds in the presence of practical battery imperfections. Numerical results demonstrate that the proposed scheme can significantly outperform the existing alternatives.

The rest of the paper is organized as follows. The system models are described in Section II. The proposed dynamic resource management scheme
is developed and analyzed in Section III. Numerical results are provided in Section IV, followed by concluding remarks in Section V.
%

\section{System Models}


As with \cite{Gat10, LHua13, Mao12}, we consider a general EH multi-hop wireless network that operates in slotted time. For convenience, the slot duration is normalized to unity; thus, the terms ``energy'' and ``power'' can be sometimes used interchangeably. The network is represented by a graph ${\cal G} = ({\cal N}, {\cal L})$, where ${\cal N}=\{1, \ldots, N\}$ denotes the set of network nodes, and ${\cal L}=\{[n,m], \; n, m \in {\cal N}\}$ collects the directed links between nodes. For each node $n \in {\cal N}$, define two sets of neighbor nodes ${\cal N}_n^o :=\{m:\;  \forall [n,m] \in {\cal L}\}$, and ${\cal N}_n^i:=\{m:\;  \forall [m,n] \in {\cal L}\}$. Further define
\begin{equation}
    d_{\max} := \max_n \{|{\cal N}_n^i|, |{\cal N}_n^o|\} \nonumber
\end{equation}
as the maximum in- and out-degree for nodes in the network.

\subsection{Network Traffic and EH Model}

The network delivers packets for data flows (called commodity in \cite{Gat10, LHua13}) indexed by their destination nodes $c$. Per time slot $t$, a data admission is implemented by the network to decide the number $R_n^c(t)$ of packets for flow $c$ that can be newly admitted at node $n$. We assume that
\begin{equation}\label{eq.Rn}
    0 \leq R_n^c(t) \leq R_{\max}, \quad \forall n, c
\end{equation}
with some finite $R_{\max} < \infty$ at all time.

Every node in the network is assumed to be capable of harvesting energy from the environmental (e.g., renewable) sources to power its transmissions. The amount of harvested energy is clearly random and intermittent over time. Let $e_n(t)$ denote the amount of harvested energy by node $n$ at time slot $t$, and let $\boldsymbol{e}(t):=[e_1(t), \ldots, e_N[t)]$ be called the {\em energy state} at $t$. We assume that $\boldsymbol{e}(t)$ takes values in some finite set and there exists $e_{\max} < \infty$ such that
\[
    0 \leq e_n(t) \leq e_{\max}, \quad \forall n, \forall t.
\]

\subsection{Transmission Model}

Per slot $t$, let $\boldsymbol{S}(t)$ denote time-varying, random {\em channel state}, which in general can be a $N$-by-$N$ matrix where the $(n,m)$ component denotes the channel condition between nodes $n$ and $m$. We assume that $\boldsymbol{S}(t)$ take values in some finite set for all time. Given $\boldsymbol{S}(t)$, the network allocates a power vector $\boldsymbol{P}(t):=[P_{[n,m]}(t), \; \forall [n,m] \in {\cal L}]$ for data transmissions over links, where $P_{[n,m]}(t)$ denotes the power allocated to node $n$ for link $[n,m]$ at time $t$. We assume that each node has a peak power constraint such that
\begin{equation}\label{eq.Pnm}
    0 \leq \sum_{m \in {\cal N}_n^o} P_{[n,m]}(t) \leq P_{\max}, \quad \forall n
\end{equation}
with a constant $P_{\max} < \infty$ at all time.

Given the channel state $\boldsymbol{S}(t)$ and the power allocation vector $\boldsymbol{P}(t)$, the transmission rate over the link $[n,m]$ is dictated by a rate-power function
\begin{equation}\label{eq.mu}
    \mu_{[n,m]}(t) = \mu_{[n,m]} (\boldsymbol{S}(t), \boldsymbol{P}(t)).
\end{equation}
For any bounded power allocation $\boldsymbol{P}(t)$, we also assume that there exists a finite constant $\mu_{\max} < \infty$ such that
\[
    \mu_{[n,m]}(t) \leq \mu_{\max}, \quad \forall [n,m] \in {\cal L}
\]
for all time under any channel state $\boldsymbol{S}(t)$. Now let $\mu_{[n,m]}^c(t)$ denote the rate allocated to the data flow $c$ over link $[n,m]$ at time $t$. It is clear that we have:
\begin{equation}\label{eq.muc}
    \sum_c \mu_{[n,m]}^c(t) \leq \mu_{[n,m]}(t), \quad \forall [n,m].
\end{equation}

Let $\boldsymbol{Q}(t):=[Q_n^c(t), \forall n,c \in {\cal N}]$ denote the data queue backlog vector in the network at time $t$, where $Q_n^c(t)$ is the amount of data for flow $c$ queued at node $n$. For the given data admission and rate allocation, we have
\begin{equation}\label{eq.Q}
\begin{split}
    Q_n^c(t+1) \leq & \left[Q_n^c(t) - \sum_{m \in {\cal N}_n^o} \mu_{[n,m]}^c(t)\right]^+ \\
     & ~~~+ \sum_{m \in {\cal N}_n^i} \mu_{[m,n]}^c(t) +R_n^c(t), \quad \forall n,c
\end{split}
\end{equation}
with $Q_n^c(0)=0$, $\forall n,c \in {\cal N}$, $Q_c^c(t)=0$, $\forall t$, and $[x]^+ := \max\{x,0\}$. Note that the inequality in (\ref{eq.Q}) is due to the fact that some nodes may not have enough packets for flow $c$ to fill the allocated rates \cite{LHua13}.

\subsection{Imperfect Battery Model}

Every node in the network has an energy storage device, i.e., battery, to save the harvested energy. Consider a practical battery with: i) a finite capacity, ii) (dis-)charging loss, and iii) energy degeneration. Let $E_{\max} \in (0, \infty)$ denote the battery capacity, $\xi \in (0,1]$ the (dis-)charging efficiency (e.g., $\xi =0.9$ means that only 90\% of the charged or discharged energy is useful while other is conversion loss)\footnote{Here we assume without loss of generality that the charging and discharging efficiency is the same.}, and $\eta \in (0,1]$ the storage efficiency (e.g., $\eta =0.9$ means that 10\% of the stored energy will be ``leaked'' over a slot, even in the absence of discharging).

We can model the battery using an energy queue. Let $E_n(t)$ denote the energy queue size, which indicates the amount of the energy left in the battery of node $n$ at time $t$; and let $\boldsymbol{E}(t):=[E_n(t), \forall n \in {\cal N}]$. As the data transmissions are powered by the harvested energy stored in the batteries, the power allocation vector $\boldsymbol{P}(t)$ must satisfy the following ``energy availability'' constraint:
\begin{equation}\label{eq.EA}
    \sum_{m \in {\cal N}_n^o} P_{[n,m]}(t) \leq \xi \eta E_n(t), \quad \forall n
\end{equation}
where the product $\xi \eta$ captures the discharging loss and energy degeneration.

Due to the energy availability constraint (\ref{eq.EA}), the energy queue $E_n(t)$ evolves according to the dynamic equation:
\begin{align}
    & E_n(t+1) = \eta E_n(t)-\frac{\sum_{m \in {\cal N}_n^o} P_{[n,m]}(t)}{\xi} + \xi e_n(t), \label{eq.E}\\
    & 0 \leq E_n(t) \leq E_{\max} \label{eq.Emax}
\end{align}
with $E_n(0)=0$, $\forall n$.

\subsection{Network Utility Maximization}

Define the time-average rate for data flow $c$ that is admitted into node $n$, as
\[
    \bar{r}_n^c = \lim_{T\rightarrow \infty} \frac{1}{T} \sum_{t=0}^{T-1} \mathbb{E}\{R_n^c(t)\}
\]
where the expectation are taken over all sources of randomness. Each flow $c$ is associated with a utility function $U_n^c(\bar{r}_n^c)$, which is assumed to be increasing, continuously differentiable, and strictly concave. Let $g_n^c$ denote the maximum first derivative of $U_n^c(r)$, and define
\[
    g_{\max} = \max_{n,c} g_n^c
\]
which is assumed to be finite.

Note that the energy state $\boldsymbol{e}(t)$ and the channel state $\boldsymbol{S}(t)$ are random processes in our model. The EH wireless network to be controlled is thus a stochastic system. The goal is to design an online resource management scheme that chooses the data admission amounts $\boldsymbol{R}(t):=[R_n^c(t), \forall n, \forall c]$, the power allocations $\boldsymbol{P}(t)=[P_{[n,m]}(t), \forall [n,m]]$, as well as the routing and scheduling decisions $\boldsymbol{\mu}(t):=[\mu_{[n,m]}^c(t),\forall [n,m],\forall c]$ per slot $t$, so as to maximize the aggregate utility of time-average data rates subject to (s. t.) network operation constraints.
Upon defining ${\cal X}:=\{\boldsymbol{R}(t), \boldsymbol{P}(t), \boldsymbol{\mu}(t), \forall t\}$, we wish to solve
\begin{equation}\label{eq.prob}
\begin{split}
   U^{opt} := & \max_{{\cal X}} \;\sum_{n,c} U_n^c(\bar{r}_n^c) \\
   & \text{s. t.} ~~~ (\ref{eq.Rn}), (\ref{eq.Pnm}), (\ref{eq.mu}), (\ref{eq.muc}), (\ref{eq.Q}), (\ref{eq.EA}), (\ref{eq.E}), (\ref{eq.Emax}),~~ \forall t.
\end{split}
\end{equation}
Note that here the constraints (\ref{eq.Rn})--(\ref{eq.Emax}) are in fact implicitly required to hold for any realization of the underlying random state $\{\boldsymbol{e}(t), \boldsymbol{S}(t)\}$ per slot $t$.

\section{Dynamic Resource Management Scheme}

The problem (\ref{eq.prob}) is a stochastic optimization problem. The problem is challenging as the optimization variables are coupled over time due to the queue dynamics and energy availability constraints in (\ref{eq.Q})--(\ref{eq.E}). Such problems usually have to be solved by dynamic programming (DP), which suffers from a curse of dimensionality and requires significant knowledge of network statistics. Differently, we next integrate and generalize the Lyapunov optimization techniques in \cite{LHua13, Qin15} to develop a low-complexity online control algorithm, which can be proven to yield a feasible  near-optimal solution for (\ref{eq.prob}) under conditions, without requiring any statistical knowledge of stochastic EH and channel processes.

\subsection{Properties of Rate-Power Function}

To start, we assume that the rate-power function in (\ref{eq.mu}) satisfies the following two properties for any given channel state $\boldsymbol{S}$:
\begin{property}
    For two power allocation vectors $\boldsymbol{P}$ and $\boldsymbol{P}'$, where $\boldsymbol{P}'$ is obtained by changing any single component $P_{[n,m]}$ to zero, we have:
    \begin{itemize}
        \item[i)] $\mu_{[n',m']}(\boldsymbol{S}, \boldsymbol{P}) \leq \mu_{[n',m']}(\boldsymbol{S}, \boldsymbol{P}')$, $\forall [n',m'] \neq [n,m]$;
        \item[ii)] $0 \leq \mu_{[n,m]}(\boldsymbol{S}, \boldsymbol{P}) - \mu_{[n,m]}(\boldsymbol{S}, \boldsymbol{P}') \leq \delta_1 P_{[n,m]}$, for a finite constant $\delta_1 \in [0, \infty)$.
    \end{itemize}
\end{property}

\begin{property}
    For two power allocation vectors $\boldsymbol{P}$ and $\boldsymbol{P}'$, where $P_{[n,m]}'=P_{[n,m]}+\frac{\Delta P}{|{\cal N}_n^o|}$, $\forall m \in {\cal N}_n^o$, and $P_{[n',m]}'=P_{[n',m]}$, $\forall n' \neq n$, we have:
    \begin{itemize}
        \item[i)] $\mu_{[n,m]}(\boldsymbol{S}, \boldsymbol{P}) \leq \mu_{[n,m]}(\boldsymbol{S}, \boldsymbol{P}')$, $\forall \forall m \in {\cal N}_n^o$;
        \item[ii)] $0 \leq \sum_{n' \neq n} \sum_{m \in {\cal N}_{n'}^o} [\mu_{[n',m]}(\boldsymbol{S}, \boldsymbol{P}) - \mu_{[n',m]}(\boldsymbol{S}, \boldsymbol{P}')] \leq \delta_2 \Delta P$, for a finite constant $\delta_2 \in [0, \infty)$.
    \end{itemize}
\end{property}

Property 1 was one of the keys for asymptotic optimality analysis of the ESA scheme in \cite{LHua13}. Here, Properties 1 and 2 will be the keys for our feasibility and optimality gap analysis in the sequel. These properties are actually satisfied by most rate-power functions, e.g., when the rate function has finite directional derivatives with respect to power, and the rates do not improve with increased interference. To see it, two intriguing examples are provided below.

{\em Example 1 (Interference Channels)}: Consider the general interference channel case. Let $h_{[n,m]}$ denote the channel coefficient from node $n$ to node $m$; i.e., $\boldsymbol{S}=[h_{[n,m]}, \forall n, m \in {\cal N}]$. The rate function is then
\[
    \mu_{[n,m]}= \log \left(1+ \frac{|h_{[n,m]}|^2P_{[n,m]}}{\sum_{[n',m'] \neq [n,m]}(|h_{[n',m]}|^2P_{[n',m']}) +\sigma^2}\right)
\]
where $\sigma^2$ is the noise variance at nodes. It can be readily derived
\begin{equation}\nonumber
\begin{cases}
    \delta_1 = \max\{\frac{|h_{[n,m]}|^2}{\sigma^2}, \forall [n,m], \;\forall \boldsymbol{S}\} \\
    \delta_2 = \max\{\frac{\sum_{n' \neq n}\sum_{m' \in {\cal N}_{n'}^o} |h_{[n,m']}|^2  }{\sigma^2}, \;\forall \boldsymbol{S}\}
\end{cases}
\end{equation}
Clearly, $\delta_1$ and $\delta_2$ are finite as long as all channel gains $|h_{[n,m]}|^2$ are bounded at all time. This is in fact also the necessary condition for $\mu_{\max} < \infty$.

{\em Example 2 (Orthogonal Channels)}: If the wireless links do not interfere with each other (e.g., when the adjacent
nodes operate on orthogonal frequency bands), then the rate function is simply
\[
    \mu_{[n,m]}= \log \left(1+ \frac{ |h_{[n,m]}|^2P_{[n,m]} } {\sigma^2}\right).
\]
We readily have
\begin{equation}\nonumber
\begin{cases}
    \delta_1 = \max\{\frac{|h_{[n,m]}|^2}{\sigma^2}, \forall [n,m], \;\forall \boldsymbol{S}\} \\
    \delta_2 =0
\end{cases}
\end{equation}
Again, $\delta_1$ is finite if all channel gains $|h_{[n,m]}|^2$ are bounded at all time.

\subsection{The Proposed Algorithm}

We assume the following two conditions for the system parameters in development of the proposed algorithm:
\begin{align}
    & \xi e_{\max} \leq (1-\eta) E_{\max}  + \frac{P_{\max}}{\xi}; \label{eq.A1} \\
    & E_{\max} \geq \frac{P_{\max}}{\xi}+ \xi e_{\max}. \label{eq.A2}
\end{align}

Condition (\ref{eq.A1}) is a necessary condition to maintain the stability of the energy queues $E_n(t)$ for every sample path. If $\xi e_{\max} > (1-\eta) E^{\max}  + \frac{P_{\max}}{\xi}$, i.e., the maximum energy arrival is deterministically greater than the largest energy departure possible, then there exists a sample path of energy queue $E_n(t)$ that grows unbounded. The condition (\ref{eq.A1}) is thus required for establishing the sample path result in the ensuing feasibility analysis. On the other hand, condition (\ref{eq.A2}) dictates that the battery capacity is large enough to accommodate the largest possible charging/discharging range. This then makes the system ``controllable'' by the proposed online scheme.\footnote{From (\ref{eq.A1}) and (\ref{eq.A2}), we require $e_{\max} \leq \frac{1}{\xi} \min\{(1-\eta) E_{\max}  + \frac{P_{\max}}{\xi},  E_{\max}- \frac{P_{\max}}{\xi}\}$. This could be ensured by ``shaping'' or ``clipping'' the energy arrivals by some external mechanisms.}

Our algorithm depends on two algorithmic parameters, namely a ``queue perturbation'' parameter $\Gamma$ and a weight parameter $V$. Derived from the feasibility requirement of the proposed algorithm (see the proof of Proposition 1 in the sequel), any pair $(V, \Gamma)$ that satisfies the following conditions can be used:
\begin{equation}\label{eq.GV}
    0 < V < V^{\max}, \quad \Gamma^{\min} \leq \Gamma \leq \Gamma^{\max}
\end{equation}
where
\begin{flalign}
    & V^{\max} := \frac{E_{\max}-\xi e_{\max} -\frac{P_{\max}}{\xi}}{\xi (\delta_1+\delta_2)g_{\max}};\label{eq.GV3}\\
    & \Gamma^{\min} := \frac{P_{\max}}{\xi \eta} + \frac{\xi}{\eta}\delta_1 g_{\max} V;  \label{eq.GV1}\\
    & \Gamma^{\max} := \frac{E_{\max}-\xi e_{\max}}{\eta} -  \frac{\xi}{\eta}\delta_2 g_{\max} V. \label{eq.GV2}
\end{flalign}
Note that the interval for $V$ in (\ref{eq.GV}) is well-defined under the condition (\ref{eq.A2}), and the interval for $\Gamma$ is valid when $V\leq V^{\max}$.

We now present the proposed algorithm:

\textbf{Initialization}: Select a pair of $(V, \Gamma)$ satisfying (\ref{eq.GV}), and a constant $\Theta = R_{\max} +d_{\max} \mu_{\max}$.

At every time slot $t$, observe states $\{\boldsymbol{e}(t), \boldsymbol{S}(t)\}$, and queues $\{\boldsymbol{Q}(t), \boldsymbol{E}(t)\}$, then determine $\boldsymbol{R}^*(t):=[R_n^{c*}(t), \forall n, c]$, $\boldsymbol{P}^*(t)=[P_{[n,m]}^*(t), \forall [n,m]]$, and $\boldsymbol{\mu}^*(t):=[\mu_{[n,m]}^{c*}(t),\forall [n,m],\forall c]$ as follows.
\begin{itemize}
  \item \textbf{Data admission}: Choose $R_n^{c*}(t)$, $\forall n,c$, to be the optimal solution of the following problem:
  \begin{equation}\label{eq.DA}
  \begin{split}
    \max_{R_n^c(t)}~ & [V U_n^c(R_n^c(t))- Q_n^c(t) R_n^c(t)]\\
    \text{s. t.} ~& ~0 \leq R_n^c(t) \leq R_{\max}
  \end{split}
  \end{equation}

  \item \textbf{Power allocation}: Define the weight of the flow $c$ over link $[n,m]$ as the ``perturbed'' queue-backpressure:
  \begin{equation}\nonumber
    W_{[n,m]}^c(t)=[Q_n^c(t)-Q_m^c(t)-\Theta]^+;
  \end{equation}
  and define the link weight as $W_{[n,m]}(t)=\max_c W_{[n,m]}^c(t)$. Choose $\boldsymbol{P}^*(t)$ to be the optimal solution of the following problem
  \begin{equation}\label{eq.PA}
  \begin{split}
    \max_{\boldsymbol{P}(t)}\; &  \sum_n\Bigl[\sum_{m \in {\cal N}_n^o} [W_{[n,m]}(t) \mu_{[n,m]}(t)]\Bigr. \\
    & ~~~~~~\Bigl.+ \frac{\eta}{\xi}(E_n(t)-\Gamma) \sum_{m \in {\cal N}_n^o}P_{[n,m]}(t)\Bigr] \\
    \text{s. t.} ~& ~ 0 \leq \sum_{m \in {\cal N}_n^o} P_{[n,m]}(t) \leq P_{\max}, \;\; \forall n
  \end{split}
  \end{equation}
  Note that $\mu_{[n,m]}(t)$ is a function of $\boldsymbol{P}(t)$. Having obtained $\boldsymbol{P}^*(t)$, the rate allocated to link $[n,m]$ is $\mu_{[n,m]}^*(t)=\mu_{[n,m]}(\boldsymbol{S}(t), \boldsymbol{P}^*(t))$.

  \item \textbf{Routing and scheduling}: For each node $n$, choose any $\check{c} \in \arg \max_c W_{[n,m]}^c(t)$. If $W_{[n,m]}^{\check{c}}(t)>0$, set
  \[
    \mu_{[n,m]}^{\check{c}*}(t)=\mu_{[n,m]}^*(t), ~~\text{and}~ \mu_{[n,m]}^{c*}(t) = 0, \; \forall c \neq \check{c}.
  \]
  This is the well-known MaxWeight matching scheduling; that is, allocate the full rate over the link $[n,m]$ to any flow that achieves the maximum positive weight over this link. Use idle-fill if needed.

  \item \textbf{Queue updates}: Update $Q_n^c(t)$ and $E_n(t)$ via (\ref{eq.Q}) and (\ref{eq.E}), respectively, based on $\boldsymbol{R}^*(t)$, $\boldsymbol{P}^*(t)$, and $\boldsymbol{\mu}^*(t)$.
\end{itemize}

\begin{remark}
Some comments are in order.
\begin{itemize}
    \item[i)] Different from the ESA algorithm in \cite{LHua13}, there is no EH admission mechanism in the proposed algorithm; the available harvested energy could be then fully capitalized on for data transmission.

    \item[ii)] The perturbed energy queue-size $E_n(t)-\Gamma$ is weighted by $\frac{\eta}{\xi}$ in the problem (\ref{eq.PA}) to determine the optimal power allocation. These weighted are used to account for the battery degeneration and discharging loss.

    \item[iii)] Solving the power allocation problem (\ref{eq.PA}) requires centralized control and can be NP-hard. For the general ad-hoc networks with interference channels, a maximal-weighted-matching based polynomial-time solver \cite{Lin06} could be used to approximate the optimal scheduling within a constant factor in a distributed manner.
\end{itemize}
\end{remark}

The proposed algorithm is an online scheme, which dynamically makes {\em instantaneous} greedy control decisions for the stochastic system under consideration, without a-priori knowledge of any statistics of the underlying random processes. We next show that the proposed scheme can also yield a feasible and asymptotically near-optimal solution for the problem (\ref{eq.prob}) of interest under conditions.

\subsection{Feasibility Guarantee}

Note that in the proposed algorithm, energy availability constraints (\ref{eq.EA}) and the bounded energy queue constraints (\ref{eq.Emax}) are ignored. It is then not clear whether the algorithm is feasible for the problem (\ref{eq.prob}). Yet, we will show that by using any pair $(V, \Gamma)$ in (\ref{eq.GV}) and $\Theta=R_{\max} +d_{\max} \mu_{\max}$, it is guaranteed that the online control policy produced by the proposed algorithm is a feasible one for (\ref{eq.prob}) under the conditions (\ref{eq.A1})--(\ref{eq.A2}).

To this end, we first show the following lemma.
\begin{lemma}
The power allocation specified by the proposed algorithm obeys: i) $\sum_{m \in {\cal N}_n^o} P_{[n,m]}^*(t)=0$, if $E_n(t) < \Gamma- \frac{\xi}{\eta} \delta_1 g_{\max} V$; and ii) $\sum_{m \in {\cal N}_n^o} P_{[n,m]}^*(t)=P_{\max}$, if $E_n(t) > \Gamma + \frac{\xi}{\eta} \delta_2 g_{\max} V$, $\forall n$.
\end{lemma}

\begin{proof}
See Appendix A.
\end{proof}

Lemma 1 reveals partial characteristics of the proposed dynamic policy. Specifically, when the energy queue at node $n$ is large enough, peak power can be afforded for its data transmissions; i.e., $\sum_{m \in {\cal N}_n^o} P_{[n,m]}^*(t)=P_{\max}$. On the other hand, when the energy queue at node $n$ is small enough, no power should be allocated; i.e., $\sum_{m \in {\cal N}_n^o} P_{[n,m]}^*(t)=0$.

Based on the structure in Lemma 1, we can then establish the following result.
\begin{proposition}
Under the conditions (\ref{eq.A1})--(\ref{eq.A2}), the proposed algorithm guarantees: i) $\sum_{m \in {\cal N}_n^o} P_{[n,m]}^*(t)=0$, if $\xi \eta E_n(t) < P_{\max}$, and ii) $0 \leq E_n(t) \leq E_{\max}$, $\forall n$, $\forall t$.
\end{proposition}

\begin{proof}
See Appendix B.
\end{proof}

\begin{remark}
Proposition 1 establishes that node $n$ allocates nonzero power over any of its outgoing links only when its energy queue-size $E_n(t) \geq \frac{P_{\max}}{\xi \eta}$ in the proposed scheme. Hence, all the power allocation decisions are feasible, i.e., the energy availability constraint (\ref{eq.EA}) is indeed {\em redundant} and can be ignored in the problem (\ref{eq.PA}). Furthermore, the energy queue-sizes $E_n(t)$ resulted from the proposed scheme are guaranteed to be bounded within $[0, E_{\max}]$. These two results together imply that the proposed algorithm with proper selection of $(V, \Gamma)$ and $\Theta$ can always yield a feasible control policy for (\ref{eq.prob}) under the conditions (\ref{eq.A1})--(\ref{eq.A2}). Note that Proposition 1 is a {\em sample path} result; i.e., it holds for every time slots under {\em arbitrary}, even non-stationary, $\{\boldsymbol{e}(t), \boldsymbol{S}(t)\}$ processes.
\end{remark}

\subsection{Optimality Gap}

To facilitate the optimality analysis, we assume for now that the random process for  $\{\boldsymbol{e}(t), \boldsymbol{S}(t)\}$ is i.i.d. over time slots.

Define $\bar{e}_n= \mathbb{E}\{e_n(t)\}$, and
\begin{align}
\bar{E}_n &:= \lim_{T\rightarrow \infty} \frac{1}{T} \sum_{t=0}^{T-1} \mathbb{E}\{E_n(t)\}; \notag \\
\bar{P}_n &:= \lim_{T\rightarrow \infty} \frac{1}{T} \sum_{t=0}^{T-1} \mathbb{E}\{\sum_{m \in {\cal N}_n^o} P_{[n,m]}(t)\} \notag
\end{align}
Since $E_n(t)$, $\sum_{m \in {\cal N}_n^o} P_{[n,m]}(t)$, and $e_n(t)$ are all bounded, it follows from (\ref{eq.E}) that
\begin{equation}\label{relax1}
    (1-\eta)\bar{E}_n = \xi \bar{e}_n - \frac{\bar{P}_n}{\xi}.
\end{equation}
As $E_n(t) \in [0, E_{\max}]$, $\forall t$, (\ref{relax1}) then implies
\begin{equation}\label{relax2}
    0 \leq \xi \bar{e}_n - \frac{\bar{P}_n}{\xi} \leq (1-\eta)E_{\max}, \quad \forall n.
\end{equation}

For the queue dynamics in (\ref{eq.Q}), we can derive similar time-average constraints \cite{LHua13}:
\begin{equation}\label{relax3}
    \bar{r}_n^c + \sum_{m \in {\cal N}_n^i} \bar{\mu}_{[m,n]}^c \leq \sum_{m \in {\cal N}_n^o} \bar{\mu}_{[n,m]}^c  \quad \forall n, c
\end{equation}
where
\[
    \bar{\mu}_{[n,m]}^c:=\lim_{T\rightarrow \infty} \frac{1}{T} \sum_{t=0}^{T-1} \mathbb{E}\{\mu_{[n,m]}^c(t)\}.
\]

Consider the following problem:
\begin{equation}\label{eq.relax}
\begin{split}
    \tilde{U}^{opt} := & \max_{{\cal X}} \;\sum_{n,c} U_n^c(\bar{r}_n^c) \\
   & \text{s. t.} ~~~ (\ref{eq.Rn}), (\ref{eq.Pnm}), (\ref{eq.mu}), (\ref{eq.muc}),~ \forall t;~ (\ref{relax2}), (\ref{relax3}).
\end{split}
\end{equation}
Note that the queue dynamic constraints (\ref{eq.E}) and (\ref{eq.Q}), which need to be performed per realization per slot $t$, are replaced by the corresponding time-average constraints (\ref{relax2}) and (\ref{relax3}), and the constraints (\ref{eq.EA}) and (\ref{eq.Emax}) are ignored. It can be shown that the problem (\ref{eq.relax}) is a relaxed version of (\ref{eq.prob}). Specifically, any feasible solution of (\ref{eq.prob}), satisfying (\ref{eq.Q})--(\ref{eq.Emax}), $\forall t$, also satisfies (\ref{relax2}) and (\ref{relax3}), due to the relevant optimization variables. It then follows that $\tilde{U}^{opt} \geq U^{opt}$.

In (\ref{eq.relax}), the optimization variables are ``decoupled'' across time slots due to the removal of time-coupling constraints (\ref{eq.Q})--(\ref{eq.Emax}). This problem has an easy-to-characterize stationary optimal control policy as formally stated in the next lemma.
\begin{lemma}
If $\{\boldsymbol{e}(t), \boldsymbol{S}(t)\}$ is i.i.d., there exists a stationary control policy ${\cal P}^{stat}$ that is a pure (possibly randomized) function of the current $\{\boldsymbol{e}(t), \boldsymbol{S}(t)\}$, while satisfying (\ref{eq.Rn})--(\ref{eq.muc}), and providing the following guarantees per $t$:
\begin{equation}\label{eq.opt}
\begin{split}
    & R_n^{c, stat}(t) = \bar{r}_n^{c, stat}, \quad \sum_{n,c} U_n^c(\bar{r}_n^{c, stat}) = \tilde{U}^{opt}, \\
    & 0 \leq \xi \bar{e}_n - \frac{\sum_{m \in {\cal N}_n^o} \mathbb{E}\{P_{[n,m]}^{stat}(t)\}}{\xi} \leq (1-\eta)E_{\max}, \; \forall n \\
    & \bar{r}_n^{c, stat} + \sum_{m \in {\cal N}_n^i} \mathbb{E}\{\mu_{[m,n]}^{c,stat}(t)\} \leq \sum_{m \in {\cal N}_n^o}\mathbb{E}\{\mu_{[n,m]}^{c,stat}(t)\},   \; \forall n, c
\end{split}
\end{equation}
where $R_n^{c, stat}(t)$, $P_{[n,m]}^{stat}(t)$, and $\mu_{[n,m]}^{c,stat}(t)$ denote the data admission, power and rate allocation by policy ${\cal P}^{stat}$, and expectations are taken over the randomization of $\{\boldsymbol{e}(t), \boldsymbol{S}(t)\}$ and (possibly) ${\cal P}^{stat}$.
\end{lemma}

\begin{proof}
The proof argument is similar to the proof of \cite[Theorem 4.5]{Nee10}; hence, it is omitted for brevity.
\end{proof}

Lemma 2 in fact holds for many non-i.i.d. scenarios as well. Generalizations to other stationary processes, or even to non-stationary processes, can be found in \cite{Nee10, Nee10-2}.

It is worth noting that (\ref{eq.opt}) not only assures that the stationary control policy ${\cal P}^{stat}$ achieves the optimal utility for (\ref{eq.relax}), but also guarantees that the admitted data rate $R_n^{c, stat}(t)$ per slot $t$ is equal to the optimal time-average rate $\bar{r}_n^{c, stat}$ (due to the stationarity of $\{\boldsymbol{e}(t), \boldsymbol{S}(t)\}$ and ${\cal P}^{stat}$). This will be important to establish the following result:
\begin{proposition}
Suppose that conditions (\ref{eq.A1})--(\ref{eq.GV}) hold, and ${\boldsymbol{e}(t), \boldsymbol{S}(t)}$ is i.i.d. over slots. Let $\bar{r}_n^{c*}(T)$ be the time-average admitted rate vector achieved by the proposed algorithm up to time $T$, i.e., $\bar{r}_n^{c*}(T)= \frac{1}{T} \sum_{t=0}^{T-1} \mathbb{E}\{R_n^{c*}(t)\}$. Then
\[
   \lim_{T\rightarrow \infty} \inf \sum_{n,c} U_n^c(\bar{r}_n^{c*}(T)) \geq U^{opt} - \frac{B}{V}
\]
where the constant
\begin{equation}\label{eq.B}
 B= N^2 B_1 + N (B_2 + B_3),
\end{equation}
with
\begin{align}
    B_1 & = 2 d_{\max}^2 \mu_{\max}^2 + \frac{1}{2} R_{\max}^2 + 2d_{\max}\mu_{\max}R_{\max}, \notag \\
    B_2 & = \frac{1}{2}\max\{[\frac{P_{\max}}{\xi}+(1-\eta)\Gamma]^2, [-\xi e_{\max}+(1-\eta)\Gamma]^2\}, \notag \\
    B_3 & = \eta(1-\eta)\max\{(E_{\max}-\Gamma)^2, \Gamma^2\},  \notag
\end{align}
and $U^{opt}$ is the optimal value of (\ref{eq.prob}) under any feasible control algorithm, even if that relies on knowing future random realizations.
\end{proposition}

\begin{proof}
See Appendix C.
\end{proof}

\begin{remark}
Proposition 2 asserts that the proposed algorithm asymptotically yields a time-average utility with an optimality gap smaller than $\frac{B}{V}$. The proposed scheme is in fact a modified version of the queue-length based stochastic optimization scheme (e.g., the ESA in \cite{LHua13}), where the ``perturbed'' queue lengths play the role of ``stochastic'' Lagrange multipliers with a dual-subgradient solver to the problem of interest. The gap $N^2 B_1/V$ is inherited from the underlying stochastic subgradient method. On the other hand, the gap $N B_2/V$ is due to the combined effect of energy-queue perturbation and battery imperfections, while the gap $N B_3/V$ is incurred by the battery degeneration.
\end{remark}

\subsection{Main Theorem}

Based on Propositions 1 and 2, it is now possible to arrive at our main result.
\begin{theorem}\label{Them:main}
Suppose that conditions (\ref{eq.A1})--(\ref{eq.GV}) hold and ${\boldsymbol{e}(t), \boldsymbol{S}(t)}$ is i.i.d. over slots. Then the proposed algorithm yields a feasible dynamic control scheme for~\eqref{eq.prob}, which has an optimality gap $\frac{B}{V}$; i.e.,
\[
    U^{opt} \geq \lim_{T\rightarrow \infty} \inf \sum_{n,c} U_n^c(\bar{r}_n^{c*}(T)) \geq U^{opt} - \frac{B}{V}
\]
where $\bar{r}_n^{c*}(T)= \frac{1}{T} \sum_{t=0}^{T-1} \mathbb{E}\{R_n^{c*}(t)\}$ and $B$ is given by (\ref{eq.B}).
\end{theorem}

\begin{remark}
Interesting comments on the minimum optimality gap with the proposed algorithm are in order.
\begin{enumerate}
\item [i)] When $\eta=1$ (i.e., without energy degeneration), the optimality gap (regret) between the proposed algorithm and the offline optimal scheme reduces to\footnote{Note that $\xi e_{\max} \leq \frac{P_{\max}}{\xi}$ for $\eta =1$ in (\ref{eq.A1}).}
	\[
		\frac{B}{V}=\frac{N^2B_1+\frac{N}{2}(\frac{P_{\max}}{\xi})^2}{V}.
	\]
	Clearly, the minimum optimality gap is given by $B/V^{\max}$, which vanishes as $V_{\max} \rightarrow \infty$. By (\ref{eq.GV3}), such an asymptotic optimality can be achieved when we have a very large battery capacity $E^{\max} \rightarrow \infty$.

\item [ii)] When $\eta\in(0,1)$, note that the constants $B_2$ and $B_3$ are in fact functions of $\Gamma$, whereas the minimum and maximum values of $\Gamma$ in (\ref{eq.GV1}) and (\ref{eq.GV2}) depend on $V$. For a given $V_{\max}$, the minimum optimality gap, $G^{\min}(V_{\max})$, can be obtained by solving following problem:
\begin{equation}
    \min_{(V, \Gamma)} \;N^2 \frac{B_1}{V}+ N [\frac{B_2(\Gamma)}{V} + \frac{B_2(\Gamma)}{V}], ~~~\text{s. t.} ~~(\ref{eq.GV}). \notag
\end{equation}
For $V \geq 0$, we know that the quadratic-over-linear functions $\frac{[\frac{P_{\max}}{\xi}+(1-\eta)\Gamma]^2}{V}$ and $\frac{[-\xi e_{\max}+(1-\eta)\Gamma]^2}{V}$ are jointly convex in $V$ and $\Gamma$ \cite{convex}. As a point-wise maximum of these two convex function, $\frac{B_2(\Gamma)}{V}$ is convex too \cite{convex}. Similarly, $\frac{B_3(\Gamma)}{V}$ is convex, while $\frac{B_1}{V}$ is clearly convex in $V$. Since the objective function is convex and the constraints in (\ref{eq.GV}) are linear, the above problem is a convex program which can be efficiently solved by general interior-point method. Note that $G^{\min}(V_{\max})$ may not decrease as $V_{\max}$ increases, or, equivalently, $E^{\max}$ increases; see discussions in \cite{Qin15}. This makes sense intuitively because for a large battery capacity, the dissipation loss due to battery imperfections will also be enlarged. In other words, asymptotic optimality of the proposed algorithm may not be achieved in the presence of battery degeneration. The smallest possible optimality gap can be numerically computed by one dimensional search over $G^{\min}(V_{\max})$ with respect to $V_{\max}$.
\end{enumerate}
\end{remark}

\begin{remark}
The results in Theorem 1 can be generalized to the more general non-i.i.d. case where the the energy state $\boldsymbol{e}(t)$ and channel state $\boldsymbol{S}(t)$ both evolve according to some finite-state irreducible, aperiodic Markov chains. Note that the feasibility of the proposed algorithm is a sample path result which does not depend on the assumption on the random processes per Remark 2; i.e., it holds also for this non-i.i.d. case. On the other hand, the performance guarantee (i.e., optimality gap result) can be obtained by applying the so-called delayed Lyapunov drift technique, i.e., the method of analyzing regenerative cycles of the Markov random processes; see \cite[Theorem 3]{LHua13} and \cite[Theorem 2]{Qin15}.
\end{remark}

\section{Numerical Tests}\label{sec:test}
In this section, numerical results are provided to evaluate the proposed algorithm.
A simple network is considered in Fig. 1, where the nodes 1-4 collect data and send data to the sink node 7 through relay nodes 5 and 6 \cite{LHua13}.

In simulations, we assume imperfect batteries at nodes, with storage efficiency  $\eta < 1$, and  (dis-)charging efficiency $\xi \leq 1$.
Table \ref{tab: param} lists the values for $d_{\max}$ (the  maximum in- and out-degree for nodes in the network), $R_{\max}$ (the maximum packets that can be newly admitted), $P_{\max}$ (the peak power), $\mu_{\max}$ (the maximum rate over all the links) and $E_{\max}$ (the battery capacity).
The utility function is selected as:
$\sum_{n,c} U_n^c(\bar{r}_n^c)=\ln(1+\bar{r}_1^7)+\ln(1+\bar{r}_2^7) +\ln(1+\bar{r}_3^7) +\ln(1+\bar{r}_4^7)$.
Suppose that all the links are independent with each other, implying $\delta_2=0$. The link state can be either good or bad with equal probability. One unit of power can deliver two packets when the link state is good, while it can be only used to transmit one packet upon bad link state.
We also assume that the harvested energy $e_n(t)$ is i.i.d. for each node; $e_n(t)$ is either $e_{\max}$ or 0 with equal probability.
As a result, we  have:
$g_{\max}=1$, $\delta_1=2$ and $\Theta=d_{\max}\mu_{\max}+R_{\max}=7$.
To satisfy the constraint (\ref{eq.A1}), we set $e_{\max}=5$, unless otherwise specified.
Each simulation is run for 1200 slots and each result is obtained by averaging over 10 independent runs.

Fig. 2 depicts the achieved utility for the proposed algorithm with different $\Gamma$, when the weight parameter $V=30$, $\xi=1$ and $\eta = 0.98, 0.97, 0.96$. It is shown that the maximum utility is achieved when the ``queue perturbation" parameter $\Gamma$ is set to  $\Gamma^{\min}$ in this case.
This is because power allocation policy (17)  tends to save energy if a large $\Gamma$ is chosen, leading to a large data queue size $Q_n(t)$.
 On the other hand, data admission policy (\ref{eq.DA}) dictates that $R_n(t)=\frac{V}{Q_n(t)}-1$, if $0 \leq \frac{V}{Q_n(t)}-1 \leq R_{\max}$, $n=1,2,3,4$. This implies that large $Q_n(t)$ reduces rate utility.
 As a result, the utility decreases as $\Gamma$ grows. For this reason, we set $\Gamma = \Gamma^{\min}$ unless otherwise specified in all the subsequent simulations.
  Note that when $\eta=0.96$ and $\Gamma=100$, we often have $P_{[n,m]}(t)=0, \forall n, \forall t$; i.e., no data are allowed to be transmitted over links. In this case, to maintain the queue stability, network utility becomes almost 0.

Fig. \ref{fig: weight} shows the network utility of the proposed algorithm with different weight parameter $V$, when $\xi=1$ and $\eta = 0.98, 0.97, 0.96$.
It is observed that the network utility decreases  after reaching the maximum value for a given $\eta$.
Clearly, the proposed algorithm does not perform well with a small $V$.
On the other hand, a large $V$ is not a good choice in the presence of battery imperfections.
To see it, we use $\Gamma^{\min}$ in (14) and set  $\mu_{[n,m]}(t) = S_{[n,m]}(t) P_{[n,m]}(t)$ in (17). ($S_{[n,m]}(t)$ is 2 with good link state, while it is 1 when the link state is bad.) Then power allocation problem (17) is a linear program, where the weight for $P_{[n,m]}(t)$ is given by $W_{[n,m]}(t)S_{[n,m]}(t)+
\eta E_n(t)-P_{\max}-\delta_1 g_{\max} V$.
Clearly, we have $P_{[n,m]}(t)=0$,  if its weight is less than zero.
Therefore, a large $V$ implies that high energy is required to render $P_{[n,m]}(t)>0$ for transmission. However, it is hard to maintain high energy in the presence of energy degeneration.
As a result, when $\eta=0.96$ and $V=50$, utility is almost 0.
The results are consist with our discussion in Remark 4. The optimality gap with the proposed scheme does not monotonically decrease with $V$ any more in the presence of battery imperfections. Determination of the best $V$ for the algorithm may refer to the solution of the problem formulated in Remark 4-ii).



\begin{table}[t]\addtolength{\tabcolsep}{1pt}
\centering
\caption{Parameters Configuration.} \label{tab: param}
    \begin{tabular}{ c|c|c|c|c}
    \hline
$d_{\max}$ &$R_{\max}$ & $P_{\max}$   &$\mu_{\max}$   & $E_{\max}$ \\ \hline
 2         & 3         &  2           & 2  	              & 160        \\ \hline
    \end{tabular}

\end{table}

\begin{figure}
\centering
\includegraphics[width=3.5in,height=1.7in]{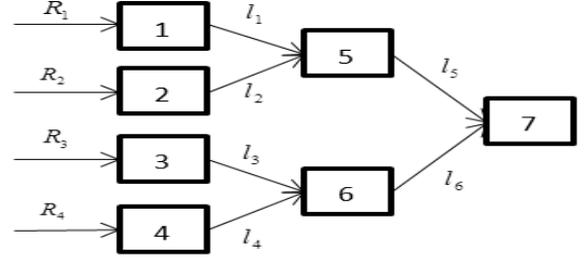}
\caption{Data collection network.}
\label{fig: network}
\end{figure}

\begin{figure}
\centering
\includegraphics[width=3.5in,height=2in]{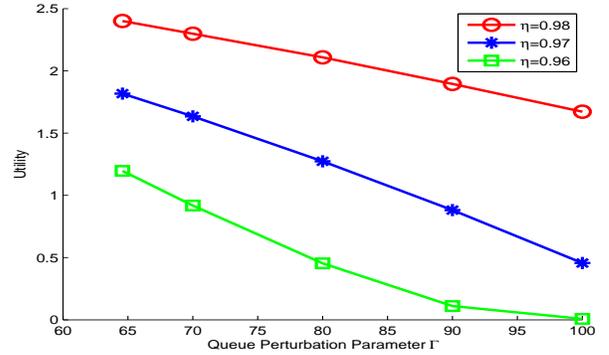}
\caption{Utility versus  ``queue perturbation" parameter $\Gamma$.}
\label{fig: gamma}
\end{figure}

\begin{figure}
\centering
\includegraphics[width=3.5in,height=2in]{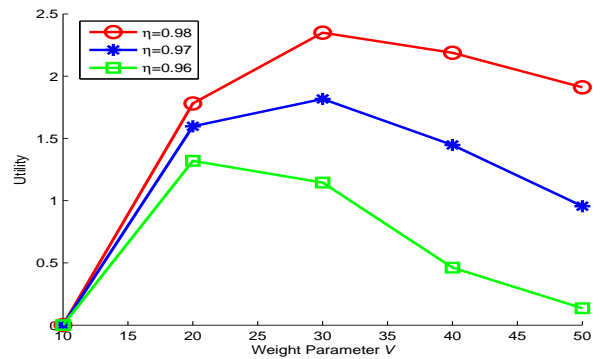}
\caption{Utility versus weight parameter $V$.}
\label{fig: weight}
\end{figure}

\begin{figure}
\centering
\subfigure{
\label{fig:subfig:a} 
\includegraphics[width=3.5in,height=2in]{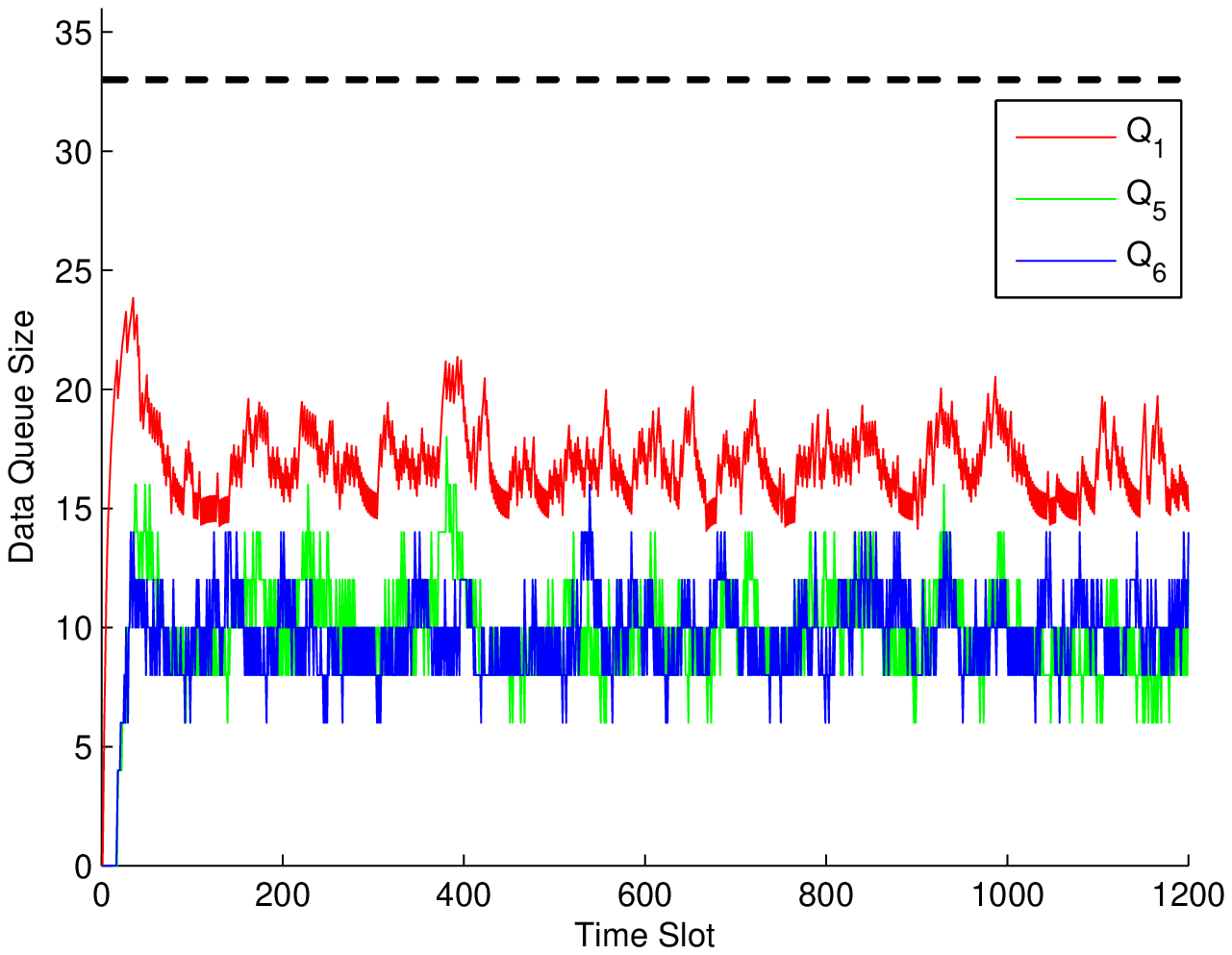}}
\subfigure{
\label{fig:subfig:b} 
\includegraphics[width=3.5in,height=2in]{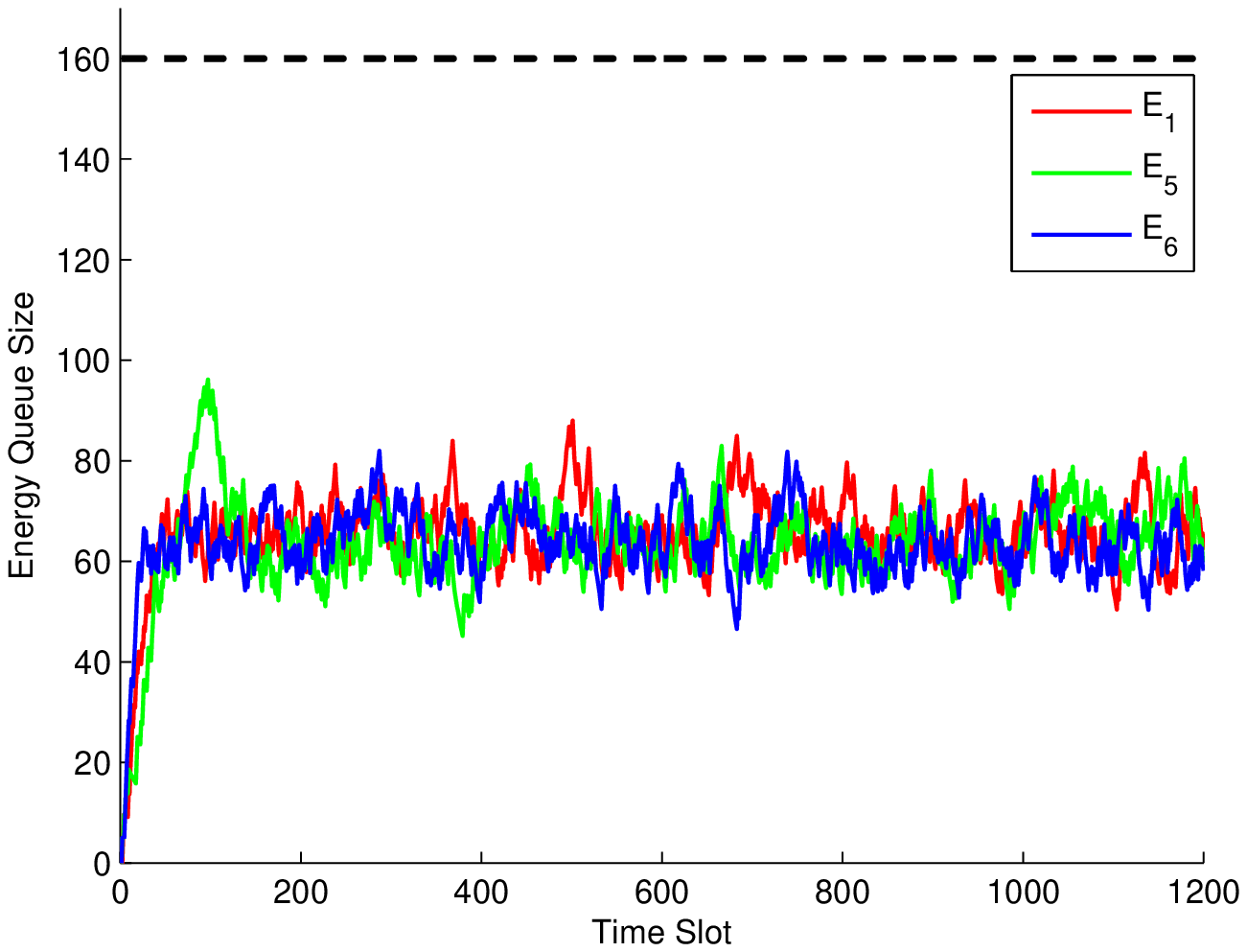}}
\caption{Sample path processes.}
\label{fig:samplePro} 
\end{figure}


\begin{figure}
\centering
\subfigure{
\label{fig:results:a} 
\includegraphics[width=3.5in,height=2in]{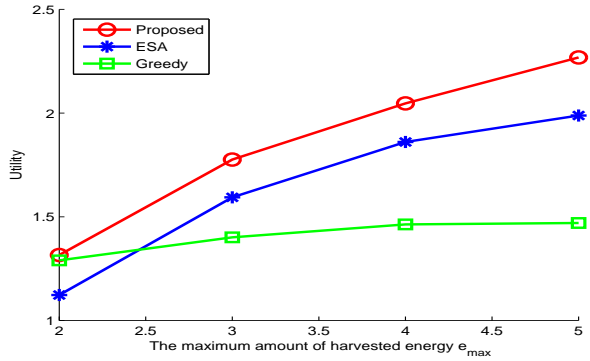}}
\caption{Comparison of the proposed algorithm, ESA and the greedy algorithm.}
\label{fig:results} 
\end{figure}

\begin{figure}
\centering
\includegraphics[width=3.5in,height=2in]{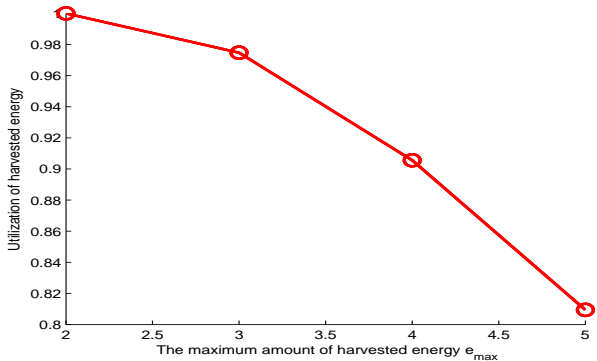}
\caption{Utilization of harvested energy for ESA.}
\label{fig: uti}
\end{figure}

The data and energy queue processes with the proposed algorithm are plotted in Fig. \ref{fig:samplePro}, when $\xi=1$, $\eta=0.98$ and $V=30$.
Clearly, the three data queue lines evolve within the region [0, $g_{\max}V+R_{\max}$]; i.e., the stability of the data network is maintained.
In addition, all energy queues satisfy $0 \leq E_n(t) \leq E_{\max}, \forall n, \forall t$.
This clearly corroborates the feasibility of the proposed scheme, as stated in Proposition 1.

Two baseline schemes are chosen to gauge the performance of the proposed algorithm.
 The first  baseline scheme is ESA in \cite{LHua13}, where perfect batteries are assumed and an EH admission is performed to ensure finite energy queues.
 The second scheme is a heuristic greedy algorithm, where a node with the largest data queue backlog is first allowed to 
 transmit data over its best outgoing links (with best link quality). The node with second largest data queue  then selects a best link to transmit, among all its outgoing links that do not conflict with the existing links.
 (A link conflicts with the existing links if it shares the same transmit or receive node with any existing links.)
 This continues until no nodes or no links can be selected. If a link $[n,m]$ is selected, the power
 $P_{[n,m]}(t)$ is set as the maximum value that satisfies both the peak power constraint (2) and the energy availability constraint (6); $\mu_{[n,m]}(t)$ is then determined correspondingly. Finally, the greedy algorithm sets $R_n(t)=R_{\max}$ if $Q_n(t)=0$, while $R_n(t)=\mu_{[n,m]}(t)$ if $Q_n(t) \neq 0$.
It is clear that the nodes have finite data queue backlogs; i.e., the system is stable. The feasibility of the greedy algorithm is also guaranteed since the energy availability constraint (6) is taken into account in power allocation policy.

Fig. \ref{fig:results} compares the performance of the three algorithms with different $e_{\max}$, when $\xi=0.95$, $\eta=0.98$ and $V=30$.
It is shown that the proposed algorithm evidently outperforms the ESA and the greedy algorithm for any given $e_{\max}$.
Since the greedy algorithm in fact schedules the links in a time division multi-access (TDMA) manner,
the resultant utility is low in general. In addition, the required energy for implementation of the greedy policy is actually also low; hence, its performance changes only slightly for different $e_{\max}$.
The proposed algorithm achieves higher utility than ESA due to two reasons.
The first reason is that ESA cannot make use of all available energy because of its EH admission mechanism, while the proposed algorithm harvested all available energy.
The utilization of available energy (the ratio of actually harvested energy to available energy) for ESA is depicted in Fig. \ref{fig: uti}.
For ESA, node $n$  stops harvesting energy when $E_n(t)$ energy reaches $\delta_1g_{\max}V+P_{\max}$ \cite{LHua13}.
Clearly, given a larger $e_{\max}$ (i.e., renewable energy is abundant), ESA  might  waste more available energy, leading to a significant performance loss.
On the other hand, ESA also has a performance loss for small $e_{\max}$ case since it does not take into account the battery imperfections.
For instance, when $e_{\max}=2$,  the utilization of available energy for ESA is 100\%; yet, the
utility of the proposed algorithm is still 17.2\% larger than that of ESA in this case.


\section{Conclusions}
In this paper, a dynamic resource allocation task was considered for general EH wireless networks. Taking into account imperfect finite-capacity energy storage devices, a stochastic optimization was formulated to maximize the long-term utility subject to the energy availability constraints. Capitalizing on Lyapunov optimization technique, an online control algorithm  was proposed to make data admission, power allocation and routing decisions, without requiring any statistical knowledge of channel, data-traffic, and EH processes. It was shown that the proposed
algorithm can efficiently utilize the harvested energy to provide a feasible and asymptotically near-optimal control solution for general EH networks.


\section*{Appendix}

\subsection{Proof of Lemma 1}

Using the perturbation weight $\Theta = R_{\max} +d_{\max} \mu_{\max}$ in queue-backpressure, it was shown in \cite[Theorem 2]{LHua13} that the data queue sizes satisfy: $0 \leq Q_n^c(t) \leq g_{\max}V + R_{\max}$, $\forall n,c$, $\forall t$. Note that this is a sample path result which holds for arbitrary $\boldsymbol{e}(t)$ and $\boldsymbol{S}(t)$ processes.

Since all the data queues are upper-bounded by $g_{\max}V + R_{\max}$, we have: $W_{[n,m]}(t) \leq [g_{\max}V -d_{\max}\mu_{\max}]^+$, $\forall [n,m]$, $\forall t$. Consider the following two cases:
\begin{itemize}
\item[c1)] If $E_n(t) < \Gamma-\frac{\xi}{\eta} \delta_1 g_{\max} V$, suppose $\sum_{m \in {\cal N}_n^o} P_{[n,m]}^*(t)> 0$; i.e., there exists a certain $P_{[n,\check{m}]}^*(t)>0$ with $\check{m} \in  {\cal N}_n^o$. We can create a new power vector $\boldsymbol{P}'$ by setting only $P_{[n,\check{m}]}'=0$ and all other components remain the same as in $\boldsymbol{P}^*(t)$. Let $G(\boldsymbol{P})$ denote the objective function of (\ref{eq.PA}). It readily follows that
    \begin{align}
        & G(\boldsymbol{P}^*(t)) -  G(\boldsymbol{P}') \notag \\
        &  = \sum_n \sum_{m \in {\cal N}_n^o} W_{[n,m]}(t)[\mu_{[n,m]}(\boldsymbol{P}^*(t)) - \mu_{[n,m]}(\boldsymbol{P}')] \notag\\
        &  ~~~~~ + \frac{\eta}{\xi} (E_n(t)-\Gamma) P_{[n,\check{m}]}^*(t) \notag \\
        &  \leq W_{[n,\check{m}]}(t)[\mu_{[n,\check{m}]}(\boldsymbol{P}^*(t)) - \mu_{[n,\check{m}]}(\boldsymbol{P}')] \notag \\
        &  ~~~~~ + \frac{\eta}{\xi} (E_n(t)-\Gamma) P_{[n,\check{m}]}^*(t) \notag \\
        &  < [g_{\max}V -d_{\max}\mu_{\max}]^+ \delta_1 P_{[n,\check{m}]}^*(t) - \delta_1 g_{\max} V P_{[n,\check{m}]}^*(t) \notag \\
        &  < 0 \notag
    \end{align}
    where the first inequality is due to Property 1-i), and the second inequality is due to Property 1-ii), as well as the bounds $W_{[n,m]}(t) \leq [g_{\max}V -d_{\max}\mu_{\max}]^+$ and $E_n(t) < \Gamma-\frac{\xi}{\eta} \delta_1 g_{\max} V$. This shows that $\boldsymbol{P}^*(t)$ cannot have been the optimal power allocation for (\ref{eq.PA}), leading to a contradiction. Hence, part i) in the lemma must hold.

\item[c2)] If $E_n(t) > \Gamma+ \frac{\xi}{\eta} \delta_2 g_{\max} V$, suppose $\sum_{m \in {\cal N}_n^o} P_{[n,m]}^*(t)<P_{\max}$. Let $\Delta P := P_{\max} - \sum_{m \in {\cal N}_n^o} P_{[n,m]}^*(t)$. We can create a new power vector $\boldsymbol{P}'$ by setting $P_{[n,m]}'=P_{[n,m]}^*(t)+\frac{\Delta P}{|{\cal N}_n^o|}$, $\forall m \in {\cal N}_n^o$, and $P_{[n',m]}'=P_{[n',m]}^*(t)$, $\forall n' \neq n$. It readily follows that
    \begin{align}
        & G(\boldsymbol{P}^*(t)) -  G(\boldsymbol{P}') \notag \\
        &  = \sum_n \sum_{m \in {\cal N}_n^o} W_{[n,m]}(t)[\mu_{[n,m]}(\boldsymbol{P}^*(t)) - \mu_{[n,m]}(\boldsymbol{P}')] \notag\\
        &  ~~~~~ - \frac{\eta}{\xi} (E_n(t)-\Gamma) \Delta P \notag \\
        &  \leq \sum_{n' \neq n} \sum_{m \in {\cal N}_{n'}^o} W_{[n',m]}(t)[\mu_{[n',m]}(\boldsymbol{P}^*(t)) - \mu_{[n',m]}(\boldsymbol{P}')] \notag \\
        &  ~~~~~ - \frac{\eta}{\xi} (E_n(t)-\Gamma) \Delta P \notag \\
        &  < [g_{\max}V -d_{\max}\mu_{\max}]^+ \delta_2 \Delta P - \delta_2 g_{\max} V \Delta P \notag \\
        &  < 0 \notag
    \end{align}
    where the first inequality is due to Property 2-i), and the second inequality is due to Property 2-ii), as well as the bounds $W_{[n,m]}(t) \leq [g_{\max}V -d_{\max}\mu_{\max}]^+$ and $E_n(t) > \Gamma+ \frac{\xi}{\eta} \delta_2 g_{\max} V$. This leads to a contradiction. Hence, part ii) in the lemma must also hold.
\end{itemize}

\subsection{Proof of Proposition 1}

Recall that $\frac{P_{\max}}{\xi \eta} = \Gamma^{\min} - \frac{\xi}{\eta} \delta_1 g_{\max} V$ by the definition of $\Gamma^{\min}$ in (\ref{eq.GV1}). When $E_n(t) < \frac{P_{\max}}{\xi \eta }$, we in fact have $E_n(t) < \Gamma^{\min} - \frac{\xi}{\eta} \delta_1 g_{\max} V \leq \Gamma - \frac{\xi}{\eta} \delta_1 g_{\max} V$. Then by Lemma 1, it readily follows $\sum_{m \in {\cal N}_n^o} P_{[n,m]}^*(t)= 0$; i.e., part i) holds.

The proof for part ii) proceeds by induction. First, recall that $E_n(0) =0$, $\forall n$, and suppose that it holds $E_n(t) \in [0, E_{\max}]$, $\forall n$, at slot $t$. We will show the bounds hold for $E_n(t+1)$, $\forall n$, as well in subsequent instances.


Note that by the definitions of $\Gamma^{\min}$ and $\Gamma^{\max}$ in (\ref{eq.GV1})--(\ref{eq.GV2}), we have $0 \leq \Gamma-\frac{\xi}{\eta} \delta_1 g_{\max} V
\leq E_{\max}$. We then consider the following three cases:
\begin{itemize}
\item[c1)] If $E_n(t) \in [0, \Gamma-\frac{\xi}{\eta} \delta_1 g_{\max} V)$, then it follows from Lemma 1 that $\sum_{m \in {\cal N}_n^o} P_{[n,m]}^*(t)= 0$. From (\ref{eq.E}), we have
    \begin{itemize}
    \item[i)] $E_n(t+1) = \eta E_n(t) + \xi e_n(t) \geq 0$;
    \item[ii)] $E_n(t+1) \leq \eta (\Gamma-\frac{\xi}{\eta} \delta_1 g_{\max} V) + \xi e_{\max} \leq E_{\max}-\xi(\delta_1+\delta_2) g_{\max} V \leq E_{\max}$, due to $V\geq 0$, $\Gamma \leq \Gamma^{\max}$, and the definition of $\Gamma^{\max}$ in (\ref{eq.GV2}).
    \end{itemize}

\item[c2)] If $E_n(t) \in [\Gamma-\frac{\xi}{\eta} \delta_1 g_{\max} V, \Gamma+ \frac{\xi}{\eta} \delta_2 g_{\max} V]$, then $0 \leq \sum_{m \in {\cal N}_n^o} P_{[n,m]}^*(t) \leq P_{\max}$. We therefore have
    \begin{itemize}
    \item[i)] $E_n(t+1) \geq \eta (\Gamma-\frac{\xi}{\eta} \delta_1 g_{\max} V) -\frac{P_{\max}}{\xi} \geq \frac{P_{\max}}{\xi}-\frac{P_{\max}}{\xi}=0$, due to $\Gamma \geq \Gamma^{\min}$ and the definition of $\Gamma^{\min}$ in (\ref{eq.GV1});
    \item[ii)] $E_n(t+1) \leq \eta (\Gamma+ \frac{\xi}{\eta} \delta_2 g_{\max} V) + \xi e_{\max} \leq E_{\max}-\xi e_{\max}  +\xi e_{\max}  = E_{\max}$, due to $\Gamma \leq \Gamma^{\max}$ as with c1-ii).
    \end{itemize}

\item[c3)] If $E_n(t) \in (\Gamma+ \frac{\xi}{\eta} \delta_2 g_{\max} V, E_{\max}]$, it follows from Lemma 1 that $\sum_{m \in {\cal N}_n^o} P_{[n,m]}^*(t) = P_{\max}$. We have
    \begin{itemize}
    \item[i)]  $E_n(t+1) \geq \eta (\Gamma+ \frac{\xi}{\eta} \delta_2 g_{\max} V) -\frac{P_{\max}}{\xi} \geq \frac{P_{\max}}{\xi} + \xi(\delta_1+\delta_2) g_{\max} V -\frac{P_{\max}}{\xi} \geq 0$, due to $\Gamma \geq \Gamma^{\min}$ as with c2-i);
    \item[ii)] $E_n(t+1) \leq \eta E_{\max} -\frac{P_{\max}}{\xi} +\xi e_{\max} \leq E_{\max}$, due to condition (\ref{eq.A1}).
    \end{itemize}
\end{itemize}

Cases c1)--c3) together prove part ii) of the proposition.

\subsection{Proof of Proposition 2}

Consider the queues under the proposed algorithm. By squaring both sides of (\ref{eq.Q}), and using the fact that $([x]^+)^2 \leq x^2$, we have
\begin{align}
    &[Q_n^c(t+1)]^2-[Q_n^c(t)]^2 \notag \\
    & \leq \left[\sum_{m \in {\cal N}_n^o} \mu_{[n,m]}^{c*}(t) \right]^2 + \left[\sum_{m \in {\cal N}_n^i} \mu_{[m,n]}^{c*}(t) + R_n^{c*}(t) \right]^2 \notag \\
    & ~~~~ -2Q_n^c(t)\left[ \sum_{m \in {\cal N}_n^o} \mu_{[n,m]}^{c*}(t) - \sum_{m \in {\cal N}_n^i} \mu_{[m,n]}^{c*}(t) - R_n^{c*}(t)\right]. \notag
\end{align}
Recall that $\sum_{m \in {\cal N}_n^o} \mu_{[n,m]}^{c*}(t) \leq d_{\max} \mu_{\max}$ and $\sum_{m \in {\cal N}_n^i} \mu_{[m,n]}^{c*}(t) + R_n^{c*}(t) \leq d_{\max} \mu_{\max}+R_{\max}$. Upon defining $\tilde{B}_1 = \frac{1}{2} d_{\max}^2 \mu_{\max}^2 + \frac{1}{2}(d_{\max}\mu_{\max}+R_{\max})^2$, we readily have:
\begin{align}
    & \frac{1}{2}([Q_n^c(t+1)]^2-[Q_n^c(t)]^2) \notag \\
    & \leq \tilde{B}_1 - Q_n^c(t)\left[ \sum_{m \in {\cal N}_n^o} \mu_{[n,m]}^{c*}(t) - \sum_{m \in {\cal N}_n^i} \mu_{[m,n]}^{c*}(t) - R_n^{c*}(t)\right]. \label{eq.dQ}
\end{align}

Similarly, we can also derive
\begin{align}
    &\frac{1}{2}([E_n(t+1)-\Gamma]^2-[E_n(t)-\Gamma]^2) \notag \\
    & \leq -\frac{1}{2}(1-\eta^2) [E_n(t)-\Gamma]^2 \notag \\
    & ~~~~ + \frac{1}{2} \left[\frac{\sum_{m \in {\cal N}_n^o} P_{[n,m]}^*(t)}{\xi} - \xi e_n(t) +(1-\eta)\Gamma\right]^2 \notag\\
    & ~~~~ -\eta(E_n(t)-\Gamma)\left[\frac{\sum_{m \in {\cal N}_n^o} P_{[n,m]}^*(t)}{\xi} - \xi e_n(t) +(1-\eta)\Gamma\right] \notag \\
    & \leq B_2 - \frac{\eta}{\xi}(E_n(t)-\Gamma) \sum_{m \in {\cal N}_n^o} P_{[n,m]}^*(t) \notag \\
    & ~~~~ + \eta (E_n(t)-\Gamma)[\xi e_n(t) - (1-\eta)\Gamma] \label{eq.dE}
\end{align}
where $B_2 = \frac{1}{2}\max\{[\frac{P_{\max}}{\xi}+(1-\eta)\Gamma]^2, [-\xi e_{\max}+(1-\eta)\Gamma]^2\}$.

Consider the Lyapunov function $L(t):=\frac{1}{2} \sum_{n,c} [Q_n^c(t)]^2 + \frac{1}{2} \sum_n [E_n(t)-\Gamma]^2$. By summing (\ref{eq.dQ}) over all $(n,c)$ and (\ref{eq.dE}) over all $n$, and by defining $\tilde{B} = N^2 \tilde{B}_1 + N B_2$, we have
\begin{align}
    & \Delta(t):=L(t+1)-L(t) \notag \\
    & \leq \tilde{B} - \sum_{n,c} Q_n^c(t)\left[ \sum_{m \in {\cal N}_n^o} \mu_{[n,m]}^{c*}(t) - \sum_{m \in {\cal N}_n^i} \mu_{[m,n]}^{c*}(t) - R_n^{c*}(t)\right] \notag \\
    & ~~~~ - \sum_n \frac{\eta}{\xi}(E_n(t)-\Gamma) \sum_{m \in {\cal N}_n^o} P_{[n,m]}^*(t)  \notag \\
    & ~~~~ + \sum_n \eta (E_n(t)-\Gamma)[\xi e_n(t) - (1-\eta)\Gamma]
\end{align}
Furthermore, we have
\begin{align}
    & \Delta(t) - V \sum_{n,c} U_n^c(R_n^{c*}(t)) + \Theta \sum_n \sum_c \sum_{m \in {\cal N}_n^o} \mu_{[n,m]}^{c*}(t) \notag \\
    & \leq \tilde{B} - \sum_{n,c} [V U_n^c(R_n^{c*}(t)) - Q_n^c(t)  R_n^{c*}(t)] \notag \\
    & ~~~~ - \sum_n \Bigl[\sum_c \sum_{m \in {\cal N}_n^o} \mu_{[n,m]}^{c*}(t)[ Q_n^c(t) -  Q_m^c(t) - \Theta]\Bigr. \notag \\
    & ~~~~~~~~~~~~~~ \Bigl. + \frac{\eta}{\xi}(E_n(t)-\Gamma) \sum_{m \in {\cal N}_n^o} P_{[n,m]}^*(t) \Bigr]  \notag \\
    & ~~~~ + \sum_n \eta (E_n(t)-\Gamma)[\xi e_n(t) - (1-\eta)\Gamma] \notag \\
    & \leq  \tilde{B} - \sum_{n,c} [V U_n^c(R_n^{c, stat}(t)) - Q_n^c(t)  R_n^{c,stat}(t)] \notag \\
    & ~~~~ - \sum_n \Bigl[\sum_c \sum_{m \in {\cal N}_n^o} \mu_{[n,m]}^{c,stat}(t)[ Q_n^c(t) -  Q_m^c(t) - \Theta]\Bigr. \notag \\
    & ~~~~~~~~~~~~~~ \Bigl. + \frac{\eta}{\xi}(E_n(t)-\Gamma) \sum_{m \in {\cal N}_n^o} P_{[n,m]}^{stat}(t) \Bigr]  \notag \\
    & ~~~~ + \sum_n \eta (E_n(t)-\Gamma)[\xi e_n(t) - (1-\eta)\Gamma] \label{eq.rhs}
\end{align}
where the second inequality is due to the fact that the proposed algorithm in fact minimizes the right-hand-side of (\ref{eq.rhs}).

Taking expectations for (\ref{eq.rhs}), we arrive at:
\begin{align}
    & \mathbb{E}\Bigl\{\Delta(t) - V \sum_{n,c} U_n^c(R_n^{c*}(t)) + \Theta \sum_n \sum_c \sum_{m \in {\cal N}_n^o} \mu_{[n,m]}^{c*}(t)\Bigr\} \notag \\
    & \leq  \tilde{B} - V \sum_{n,c} U_n^c(\bar{r}_n^{c, stat}(t)) \notag \\
    & ~~ + Q_n^c(t) \Bigl[ \bar{r}_n^{c, stat} + \sum_{m \in {\cal N}_n^i} \mathbb{E}\{\mu_{[m,n]}^{c,stat}(t)\} - \sum_{m \in {\cal N}_n^o}\mathbb{E}\{\mu_{[n,m]}^{c,stat}(t)\} \Big] \notag \\
    & ~~ + \sum_n \eta (E_n(t)-\Gamma) \Bigl[\xi \bar{e}_n - \frac{\sum_{m \in {\cal N}_n^o} P_{[n,m]}^{stat}(t)}{\xi} - (1-\eta)\Gamma \Bigr]  \notag \\
    & ~~ + \Theta \sum_n \sum_c \sum_{m \in {\cal N}_n^o} \mathbb{E}\{\mu_{[n,m]}^{c,stat}(t)\} \notag \\
    & \leq \tilde{B}  - V \tilde{U}^{opt} + N \eta(1-\eta)\max\{(E_{\max}-\Gamma)^2, \Gamma^2\} \notag \\
    & ~~ + N^2 \Theta d_{\max}\mu_{\max}
\end{align}
where the second inequality is due to (\ref{eq.opt}) in Lemma 2, $E_n(t) \in [0, E_{\max}]$ by Proposition 1, and $\sum_{m \in {\cal N}_n^o} \mathbb{E}\{\mu_{[n,m]}^{c,stat}(t)\} \leq d_{\max}\mu_{\max}$.

With $B = \tilde{B} + N \eta(1-\eta)\max\{(E_{\max}-\Gamma)^2, \Gamma^2\}+ N^2 \Theta d_{\max}\mu_{\max}$ and noticing $\mathbb{E}\{\mu_{[n,m]}^{c*}(t)\} \geq 0$, we then have
\[
    \mathbb{E}\{\Delta(t)\} - V \sum_{n,c} \mathbb{E}\{U_n^c(R_n^{c*}(t))\} \leq B  - V \tilde{U}^{opt} \leq B  - V U^{opt}
\]
Summing over all $t=1,2,\ldots, T-1$, we have
\begin{equation}
    \mathbb{E}\{L(T)\}-L(0) -V \sum_{t=0}^{T-1} \sum_{n,c}
    \mathbb{E}\{U_n^c(R_n^{c*}(t))\} \leq T(B-VU^{opt}) \nonumber
\end{equation}
which leads to
\begin{align}
    &\frac{1}{T} \sum_{t=0}^{T-1} \sum_{n,c} \mathbb{E}\{U_n^c(R_n^{c*}(t))\} \geq U^{opt} - \frac{B}{V} - \frac{L(0)}{VT}.
\end{align}
Using Jensen's inequality, we see that
\[
    \sum_{n,c} U_n^c(\frac{1}{T} \sum_{t=0}^{T-1} \mathbb{E}\{R_n^{c*}(t\}) \geq U^{opt} - \frac{B}{V} - \frac{L(0)}{VT}.
\]
The proposition follows by taking the limit $T \rightarrow \infty$ and using the definition of $\bar{r}_n^{c*}(T)$.

\end{document}